\documentclass[prd,showpacs,preprintnumbers,groupedaddress]{revtex4-1} %for arXiv
\usepackage{amsthm}
\theoremstyle{plain}
\newtheorem{theorem}{Theorem}
\newtheorem{definition}{Definition}
\newtheorem*{lemma}{Lemma}

\newtheorem{proposition}{Proposition}
\usepackage{cases}
\usepackage{color}
\usepackage{comment}
\usepackage{ulem}
\usepackage{amsfonts}

\newcommand{\llaabel}[1]{\label{#1}}

\begin{document}

%\preprint{}

\title{Photon surfaces in spherically, planar and hyperbolically symmetric spacetimes in D-dimensions: Sonic point/photon sphere correspondence}


\author{Yasutaka Koga}
%\thanks{}
%\author{}
%\thanks{}
\affiliation{Department of Physics, Rikkyo University, Toshima, Tokyo 171-8501, Japan}


\date{\today}

\begin{abstract}
Sonic point/photon sphere (SP/PS) correspondence is a theoretical phenomenon which appears in fluid dynamics on curved spacetime and its existence has been recently proved in quite wide situations as theorems.
The theorems state that a sonic point (SP) of radiation fluid flow must be on an unstable photon sphere (PS) when the fluid flows radially or rotationally on an equatorial plane in spherically symmetric spacetime of arbitrary dimensions.
In this paper, we investigate SP/PS correspondence in spherically, planar and hyperbolically symmetric spacetime.
As the corresponding objects of photon spheres in nonspherically symmetric spacetime, we consider photon surfaces introduced by Claudel {\it et al.}~(2001) in the spacetime.
After formulating the problem of radial fluid flows, we prove there always exists a correspondence between the sonic points and the photon surfaces, namely, SP/PS correspondence in nonspherically symmetric spacetime.
\end{abstract}

%\keywords{}

\maketitle

%\section{}
%\subsection{}
%\subsubsection{}

\tableofcontents

%%%sec
\section{Introduction}
\llaabel{sec:introduction}
A photon sphere is a sphere of spacetime on which null geodesics take circular orbits.
In astrophysical cases, black holes usually have photon spheres near their horizons.
A photon sphere has been widely studied in its various aspects;
For optical observations of black holes through background light emission, the photon sphere determines the size of the black hole shadow.
In the case of the Schwarzschild black hole for example, we can see their relation from the calculation  by Synge~\cite{synge};
Properties of gravitational waves from black holes are also closely related to the photon sphere.
It is known that the frequencies of quasinormal modes are determined by
the parameters of null geodesic motions on and near the photon sphere in
various situations~\cite{cardoso}~\cite{hod}.
The nature of photon sphere itself and the generalization to nonspherically symmetric spacetimes were also investigated by Claudel, Virbhadra and Ellis~\cite{claudel} and the citations.
\par
Accretion of fluid onto objects is a basic problem in astrophysics and has been investigated as fluid dynamics on curved spacetime in general relativistic contexts.
The accretion problem has been widely studied in the cases of Newtonian gravity~\cite{bondi}, Schwarzschild spacetime~\cite{michel}, Schwarzschild (anti-)de Sitter spacetime~\cite{mach} and generic spherically symmetric spacetime~\cite{chaverra}.
One of the interesting features is the existence of transonic flow and its sonic point (or critical point), that is, flow which transits from subsonic to supersonic state and its transition point.
A sonic point generally appears in accretion problems and plays a key role in the analysis.
\par
{\it Sonic point/photon sphere(SP/PS) correspondence} is a theoretical phenomenon which appears in fluid dynamics on curved spacetime.
It refers to a coincidence of radii of a sonic point (SP) of radiation fluid flow and a photon sphere (PS).
We can easily observe that SP/PS correspondence holds in Schwarzschild (anti) de-Sitter spacetime from the work by Mach {\it et al.}~\cite{mach} with the knowledge that the spacetime has a photon sphere on the radius $3M$.
Surprisingly, it was recently found that SP/PS correspondence always holds in arbitrary static spherically symmetric spacetime of arbitrary dimensions for spherical flow~\cite{koga} and rotational flow~\cite{koga2}.
In addition to its applications to observations in astrophysics, SP/PS correspondence has interesting implications concerning fundamental mechanism of fluid dynamics on curved spacetime.
This is because a sonic point of radiation fluid flow can be regarded as a point characterized by macroscopic behavior of a system of photons and a photon sphere is characterized by microscopic behavior of a single photon.
The two physics can be closely related at a deeper level in this correspondence.
\par
The existence of SP/PS correspondence in quite wide situations~\cite{koga}~\cite{koga2} strongly suggests that there exists some physical reason for it.
To reveal the reason, it is necessary to know what properties of a photon sphere are needed for the correspondence to hold; Circularity of null orbits, positivity and constancy of the intrinsic curvature or else.
Claudel {\it et al.}~\cite{claudel} introduced a geometrical concept, {\it photon surface}, which inherits only the local geometrical property of photon spheres called {\it umbilicity} but need not have spherical symmetry.
In this paper, we see there exists the correspondence between sonic points and {\it photon surfaces} in nonspherical spacetime like as SP/PS correspondence.
The result leads us to the conclusion that {\it SP/PS correspondence is caused by the umbilicity of a photon sphere}.
\par
We consider static spacetime of spherical, planar and hyperbolic symmetry given by the metric,
\begin{equation}
\llaabel{eq:metric}
ds^2=-f(r)dt^2+g(r)dr^2+r^2\left(d\chi^2+s^2(\chi)d\Omega^2_{D-3}\right),
\end{equation}
where $f(r)>0$, $g(r)>0$ and the function $s(\chi)$ is given by
\begin{equation}
\llaabel{eq:maxsym2space}
s(\chi)=
	\left\{
	\begin{array}{ll}
	\sin \chi& (spherical)\\
	\chi& (planar)\\
	\sinh \chi& (hyperbolic)
	\end{array}
	\right.
\end{equation}
in the spherically, planar and hyperbolically symmetric case, respectively.
$d\Omega^2_{D-3}$ is a unit $(D-3)$-sphere,
\begin{equation}
d\Omega_{D-3}^2=d\theta_1^2+\cdots+\sin^2\theta_1\cdots\sin^2\theta_{D-5}d\theta_{D-4}^2
+\sin^2\theta_1\cdots\sin^2\theta_{D-4}d\theta_{D-3}^2.
\end{equation}
We investigate photon surfaces of $r=const.$ hypersurfaces, which we call constant-$r$ photon surfaces, in the spacetime.
(The photon surfaces in the spherical case are just photon spheres.)
Then, after formulating an accretion problem of radial fluid flows for general equation of state (EOS), we prove our main theorem, which states that there exists the correspondence between the sonic points of the radiation fluid flows and the photon surfaces:
\begin{theorem}
\llaabel{theorem:correspondence}
For any physical transonic flow of radiation fluid which is stationary and spherically, planar or hyperbolically symmetric on the spacetime~(\ref{eq:metric}), the radius of its sonic point coincides with that of (one of) the unstable constant-$r$ photon surface(s).
\end{theorem}
\par
In Sec.~\ref{sec:psf-arbitrary-dim}, we review the works by Claudel {\it et al.}~\cite{claudel} and Perlick~\cite{perlick} and find the condition for hypersurfaces to be photon surfaces.
Using the result, we investigate photon surfaces of constant radius in the spacetime~(\ref{eq:metric}) in Sec.~\ref{sec:psf-const-r}.
We define the stability of the photon surfaces of constant radius which is analogous to the stability of the photon sphere defined by Koga and Harada~\cite{koga} in Sec.~\ref{sec:stability}.
Then we formulate the accretion problem of radial fluid flow in Sec.~\ref{sec:accretion} and finally prove the correspondence between the sonic point of radiation fluid flow and the photon surfaces.
The conclusion is given in Sec.~\ref{sec:conclusion}.

%%%sec
\section{Photon surface of arbitrary dimensions}
\llaabel{sec:psf-arbitrary-dim}
A photon surface is a geometrical structure of spacetime first introduced by Claudel {\it et al.}~\cite{claudel}.
This is one of generalizations of a photon sphere and can be defined for any spacetime, $(M, {\bf g})$, even if the spacetime has no symmetries:
\begin{definition}[Photon surface]
\label{definition:photonsurface}
A photon surface of $(M, {\bf g})$ is an immersed, nowhere-spacelike
hypersurface $S$ of $(M, {\bf g})$ such that, for every point $p\in S$ and every null vector ${\bf k}\in T_pS$, there exists a null geodesic $\gamma : (-\epsilon,\epsilon) \to M$ of $(M, {\bf g})$ such that $\gamma(0) ={\bf k}, |\gamma|\subset S$.
\end{definition}
A photon sphere is a timelike hypersurface $\mathbb{R}\times S^2$ along which all the null geodesics take circular orbits.
A photon surface inherits not the symmetries and the global properties of photon sphere, such as the spatial topologies of the surface and the orbits on it, but only the local properties.
Claudel {\it et al.}~\cite{claudel} proved that the following theorem holds for a 3-dimensional timelike photon surface in $4$-dimensional spacetime (Theorem~2.2 in Sec.~2 in~\cite{claudel}):
\begin{theorem}[Equivalent conditions for photon surface]
\label{theorem:equivalent}
Let $S$ be a timelike hypersurface of $(M, {\bf g})$.
Let ${\bf n}$ be a unit normal
field to $S$ and let $h_{ab}$ be the induced metric on $S$. Let $\chi_{ab}$ be the second fundamental
form on $S$ and let $\sigma_{ab}$ be the tracefree part of $\chi_{ab}$.
Then the following are equivalent:
\begin{enumerate}
\renewcommand{\labelenumi}{\roman{enumi})}
\item $S$ is a photon surface;
\item $\chi_{ab}k^ak^b=0$\ $\forall$ null ${\bf k}\in T_pS\ \forall p\in S$;
\item $\sigma_{ab}=0$;
\item every affine null geodesic of $(S,{\bf h})$ is an affine null geodesic of $(M,{\bf g})$.
\end{enumerate}
\end{theorem}
To find timelike photon surfaces of a given spacetime, it is easier to investigate whether a given timelike hypersuface $S$ is umbilic (i.e. the second fundamental form is pure-trace) everywhere on $S$ or not.
\par
Perlick~\cite{perlick} proved propositions similar to Theorem~\ref{theorem:equivalent} for a surface of arbitrary codimensions and spacetime of arbitrary dimensions.
The one of his propositions is for a Lorentzian manifold $M$ (Proposition~3 in Sec.~3 in~\cite{perlick}):
\begin{proposition}
\llaabel{proposition:D-DimEquivalent}
Let $\tilde{M}$ be a timelike submanifold with $2\le\dim(\tilde{M})\le\dim\left(M\right)$.
Then $\tilde{M}$ is totally umbilic if and only if every lightlike geodesic that starts tangent to $\tilde{M}$ remains within $\tilde{M}$ (for some parameter interval around the starting point).
\end{proposition}
The proposition states that a timelike submanifold $\tilde{M}$ is a photon surface if and only if it is totally umbilic.
(Strictly speaking, photon surface is originally defined as a hypersurface, i.e. a surface of codimension one in~\cite{claudel}.
However, we can easily generalize Definition~\ref{definition:photonsurface} for surfaces of arbitrary codimensions.)
His work includes the definition of {\it totally umbilic} submanifold which is applicable to both the degenerate and nondegenerate submanifold.
However, for a nondegenerate submanifold of codimension one, such as a timelike hypersurface, the total umbilicity is equivalent to that the second fundamental form is pure-trace everywhere on the submanifold, i.e. $\sigma_{ab}=0$ $\forall p\in \tilde{M}$.
\par
From Theorem~\ref{theorem:equivalent} and Proposition~\ref{proposition:D-DimEquivalent}, we have the following proposition which we need for our purpose:
\begin{proposition}
\label{proposition:D-dimPhotonsurface}
A timelike hypersurface $S$ of $(M,g)$ is a photon surface if and only if it is totally umbilic.
It is equivalent to the condition,
\begin{equation}
\llaabel{eq:umbilic}
\sigma_{ab}=0\;\forall p\in S.
\end{equation}
\end{proposition}
\par
In the next section, we investigate the explicit condition for a specific timelike hypersurface to be a photon surface in the spacetime we interested in by calculating Eq.~(\ref{eq:umbilic}).

%%%sec
\section{Photon surface of constant radius}
\llaabel{sec:psf-const-r}
We investigate the explicit condition for a timelike photon surface to be a photon surface in the spherically, planar and hyperbolically symmetric spacetime of $D$-dimnesions given by the metric~(\ref{eq:metric}).
Here we focus our attention on a timelike photon surface of constant radius.
We consider the [$(D-1)$-dimensional] timelike hypersurface of constant radius $S_r$ defined by
\begin{equation}
S_r:=\left\{p\in M|r=const. \right\}
\end{equation}
and investigate the condition for $S_r$ to be a photon surface.
Note that, in the spherical case, the photon surface exactly coincides with the photon sphere studied in~\cite{koga}.
%%%Subsec
\subsection{Second fundamental form}
The hypersurface $S_r$ has the normal,
\begin{equation}
n_a=\sqrt{g}dr_a,
\end{equation}
and the induced metric $h_{ab}$ on it,
\begin{equation}
h_{ab}=g_{ab}-n_an_b.
\end{equation}
\par
The second fundamental form $\chi_{ab}$ is given by
\begin{eqnarray}
\chi_{ab}&:=&h^c_a\nabla_cn_b\nonumber\\
&=&h^c_a(\partial_cn_b-{\Gamma^d}_{cb}n_d)\nonumber\\
&=&-h^c_a{\Gamma^d}_{cb}n_d
\end{eqnarray}
where we used the fact that $h_\mu^\nu\partial_\nu=(\partial_t,0,\partial_\chi,\partial_{\theta_1},...,\partial_{\theta_{D-3}})$.
The components are obtained from
\begin{eqnarray}
\chi_{\mu\nu}
&=&-h^\sigma_\mu{\Gamma^\rho}_{\sigma\nu}n_{\rho}\nonumber\\
&=&-{\Gamma^r}_{\mu\nu}\sqrt{g}
\end{eqnarray}
where $\mu,\nu=t,\chi,\theta_1,...,\theta_{D-3}$.
Hereafter we calculate the components in the tetrad system $\{e_{(\mu)}\}$ defined so that
\begin{equation}
e_{(\mu)}\propto\partial_\mu.
\end{equation}
Because the brackets of the third term in Eq.~(\ref{eq:metric}) represents constant curvature $(D-2)$-space, it is sufficient to evaluate the components $\chi_{\left(\mu\right)\left(\nu\right)}$ only for $\mu,\nu=t,\chi,\theta$ where $\theta:=\theta_1$ corresponds to one of the coordinate of the $(D-3)$-sphere $d\Omega^2_{D-3}$.
The calculations of the Christoffel symbol (See Appendix~\ref{app:christoffel}) gives
%\red{(The correction of Christoffel not yet applied to.)}
\begin{equation}
\chi_{(i)(j)}=\sqrt{g}^{-1}diag\left[-\frac{1}{2}\frac{f'}{f},\frac{1}{r},\frac{1}{r}\right]
\end{equation}
where $i,j=t,\chi,\theta$.
\par
The trace $\Theta$ of the second fundamental form is given by
\begin{equation}
\Theta:=h^{ab}\chi_{ab}=\eta^{(\mu)(\nu)}\chi_{(\mu)(\nu)}
=\sqrt{g}^{-1}\left[\frac{1}{2}\frac{f'}{f}+\left(D-2\right)\frac{1}{r}\right].
\end{equation}
\par
The trace-free part $\sigma_{ab}$ of the second fundamental form is then given by
\begin{equation}
\sigma_{ab}=\chi_{ab}-\frac{1}{D-1}\Theta h_{ab}.
\end{equation}
Its components are given by
\begin{equation}
\sigma_{(i)(j)}=-\frac{1}{2(D-1)}\frac{(fr^{-2})'}{(fr^{-2})}\sqrt{g}^{-1}diag\left[D-2,1,1\right].
\end{equation}
Finally, for all the components including $\mu,\nu=r$, we have
\begin{equation}
\label{eq:sigma}
\sigma_{(\mu)(\nu)}=-\frac{1}{2(D-1)}\frac{(fr^{-2})'}{(fr^{-2})}\sqrt{g}^{-1}diag\left[D-2,0,1,...,1\right].
\end{equation}
%%%Subsec
\subsection{Condition for photon surface}
From Proposition~\ref{proposition:D-dimPhotonsurface}, the timelike hypersurface $S_r$ is a photon surface if and only if it is totally umbilic, i.e. $\sigma_{ab}=0\; \forall p\in S$.
Then we obtain the following proposition from Eq.~(\ref{eq:sigma}):
\begin{proposition}
\label{proposition:const-r-psf}
A timelike hypersurface $S_r$ of the radius $r$ is a photon surface if and only if
\begin{equation}
\label{eq:const-r-psf}
(fr^{-2})'=0
\end{equation}
is satisfied at the radius.
\end{proposition}
We hereafter call a hypersurface of constant radius, $S_r$, {\it constant-$r$ photon surface} if it is a photon surface.
%We hereafter call the timelike photon surface, $S_r$, {\it constant-$r$ photon surface}.

%%%Sec
\section{Stability of constant-$r$ photon surface}
\llaabel{sec:stability}
A photon sphere can be applied to many physics, for example, BH shadows, quasinormal modes and gravitational instability of spacetime.
The important thing which plays a crucial role in the applications is {\it stability} of a photon sphere.
BH shadows, for example, are shaped by unstable photon spheres.
Here we define {\it stability} of a constant-$r$ photon surface like as a photon sphere and find its condition.
%%%Subsec
\subsection{Definition}
A photon sphere (in spherically symmetric spacetime) can be classified into a stable and unstable one in the sense that the circular null geodesics along the sphere take stable or unstable circular orbits.
A circular orbit is a orbit of constant radius and is said to be stable if the orbit is stable against small radial perturbation and otherwise unstable.
In the planar and hyperbolic case of spacetime~(\ref{eq:metric}), we here say a orbit of constant radius, $r$, is stable or unstable in that sense.
Then we define the {\it stability} of the constant-$r$ photon surface as in the same way as the stability of photon sphere defined in~\cite{koga}:
\begin{definition}
\llaabel{definition:stability}
Let $S_r$ be a timelike constant-$r$ photon surface.
Let $\gamma$ be a null geodesic along $S_r$ and therefore a orbit of constant radius.
Then $S_r$ is said to be stable if every $\gamma$ is a stable orbit and unstable if every $\gamma$ is an unstable orbit.
\end{definition}
Note that if one finds some $\gamma$ which is stable (unstable), it implies every $\gamma$ along $S_r$ is stable (unstable) and therefore the constant-$r$ photon surface $S_r$ is also stable (unstable) because the static slice of each $r=const.$ hypersurfaces has maximal symmetry.
%%%Subsec
\subsection{Condition}
For the stability condition of the constant-$r$ photon surface, the following proposition holds:
\begin{proposition}
\llaabel{proposition:stability}
The stability condition of the constant-$r$ photon surface $S_r$ is given by
\begin{equation}
\label{eq:stability}
stable\,(unstable) \Leftrightarrow (fr^{-2})''>0\, (<0)
\end{equation}
at the radius.
\end{proposition}
\begin{proof}
%The proof is given in the following.
As mentioned below Definition~\ref{definition:stability}, it is sufficient to see stability of only one null geodesic along the photon surface.
\par
Stability of orbits of constant radius can be analyzed by means of effective potentials.
Without loss of generality, we focus on null geodesics on the equatorial plane of the unit $(D-3)$-sphere, $d\Omega_{D-3}^2$.
Then the null geodesics are effectively on the $4$-dimensional spacetime,
\begin{equation}
ds^2=-f(r)dt^2+g(r)dr^2+r^2\left(d\chi^2+s^2(\chi)d\phi^2\right),
\end{equation}
where $\phi:=\theta_{D-3}$.
For all the cases of $s(\chi)$ in Eq.~(\ref{eq:maxsym2space}), we have a timelike Killing vector $\xi_t:=\partial_t$, three spatial Killing vectors $\xi_i\,(i=1,2,3)$ and the corresponding conserved quantity, the energy $E:=-g_{\mu\nu}\dot{x}^\mu\xi_t^\nu$ and the three angular momentums $L_{i}:=g_{\mu\nu}\dot{x}^\mu\xi_i^\nu\,(i=1,2.3)$.
Then the Hamiltonian of the null geodesics,
\begin{equation}
\mathcal{H}:=\frac{1}{2}g_{\mu\nu}\dot{x}^\mu\dot{x}^\nu\left(=0\right),
\end{equation}
reduces to
\begin{equation}
\label{eq:hamiltonian}
\mathcal{H}=\frac{1}{2}\left[-f(r)\dot{t}^2+g(r)\dot{r}^2+r^2\left(\dot{\chi}^2+s^2(\chi)\dot{\phi}^2\right)\right]
\end{equation}
where the dots denotes the derivatives by the affine parameter $\lambda$, $\dot{}:=d/d\lambda$.
From the symmetry of the spacetime, we can further reduce the problem to the one with the only two parameters, energy, $E$, and the one of the angular momentums, $L$, by an appropriate choice of the initial value in each case of the spatial symmetry as follows.
%%%subsubsec
%\subsubsection{Spherical case}
\par
In the spherically symmetric case, we have the spatial Killing vectors,
\begin{eqnarray}
\xi_1&=&\sin\phi\partial_\chi+\cot\chi\cos\phi\partial_\phi,\nonumber\\
\xi_2&=&\cos\phi\partial_\chi-\cot\chi\sin\phi\partial_\phi,\nonumber\\
\xi_3&=&\partial_\phi.
\end{eqnarray}
The angular momentums are
\begin{eqnarray}
L_1&=&\sin\phi p_\chi+\cot\chi\cos\phi p_\phi,\nonumber\\
L_2&=&\cos\phi p_\chi-\cot\chi\sin\phi p_\phi,\nonumber\\
L_3&=& p_\phi
\end{eqnarray}
where $ p_\mu:=g_{\mu\nu}\dot{x}^\nu$.
Taking the initial value of the orbit as $\phi=\pi/2$ and $\dot{\phi}=0$, we have $L_1=r^2\dot{\chi}$ and $L_2=L_3=0$ for all time.
The Hamiltonian Eq.~(\ref{eq:hamiltonian}) reduces to
\begin{equation}
\label{eq:hamiltonian-spherical}
\mathcal{H}=g\left[\frac{1}{2}\dot{r}^2-\frac{1}{2fg}\left(E^2-fr^{-2}L_1^2\right)\right].
\end{equation}
%%%subsubsec
%\subsubsection{Planar case}
\par
In the planar case, we have the spatial Killing vectors,
\begin{eqnarray}
\xi_1&=&\sin\phi\partial_\chi+\frac{1}{\chi}\cos\phi\partial_\phi,\nonumber\\
\xi_2&=&\cos\phi\partial_\chi-\frac{1}{\chi}\sin\phi\partial_\phi,\nonumber\\
\xi_3&=&\partial_\phi.
\end{eqnarray}
The angular momentums are
\begin{eqnarray}
L_1&=&\sin\phi p_\chi+\frac{1}{\chi}\cos\phi p_\phi,\nonumber\\
L_2&=&\cos\phi p_\chi-\frac{1}{\chi}\sin\phi p_\phi,\nonumber\\
L_3&=&p_\phi.
\end{eqnarray}
Taking the initial value of the orbit as $\phi=\pi/2$ and $\dot{\phi}=0$, we have $L_1=r^2\dot{\chi}$ and $L_2=L_3=0$ for all time.
The Hamiltonian Eq.~(\ref{eq:hamiltonian}) reduces to
\begin{equation}
\label{eq:hamiltonian-planar}
\mathcal{H}=g\left[\frac{1}{2}\dot{r}^2-\frac{1}{2fg}\left(E^2-fr^{-2}L_1^2\right)\right].
\end{equation}
%%%subsubsec
%\subsubsection{Hyperbolic case}
\par
In the hyperbolically symmetric case, we have the spatial Killing vectors,
\begin{eqnarray}
\xi_1&=&\sin\phi\partial_\chi+\coth\chi\cos\phi\partial_\phi,\nonumber\\
\xi_2&=&\cos\phi\partial_\chi-\coth\chi\sin\phi\partial_\phi,\nonumber\\
\xi_3&=&\partial_\phi.
\end{eqnarray}
The angular momentums are
\begin{eqnarray}
L_1&=&\sin\phi p_\chi+\coth\chi\cos\phi p_\phi,\nonumber\\
L_2&=&\cos\phi p_\chi-\coth\chi\sin\phi p_\phi,\nonumber\\
L_3&=& p_\phi.
\end{eqnarray}
Taking the initial value of the orbit as $\phi=\pi/2$ and $\dot{\phi}=0$, we have $L_1=r^2\dot{\chi}$ and $L_2=L_3=0$ for all time.
The Hamiltonian Eq.~(\ref{eq:hamiltonian}) reduces to
\begin{equation}
\label{eq:hamiltonian-hyperbolic}
\mathcal{H}=g\left[\frac{1}{2}\dot{r}^2-\frac{1}{2fg}\left(E^2-fr^{-2}L_1^2\right)\right].
\end{equation}
\par
As a result, from Eqs.~(\ref{eq:hamiltonian-spherical}),~(\ref{eq:hamiltonian-planar}) and~(\ref{eq:hamiltonian-hyperbolic}), we obtain the same one-dimensional effective equation of motion,
\begin{equation}
\frac{1}{2}\dot{r}^2+V(r)=0,\;V(r):=-\frac{1}{2fg}\left(E^2-fr^{-2}L^2\right),
\end{equation}
for all the cases of symmetry where $L:=L_1$ is the angular momentum in each of the cases.
The orbit of constant radius along the constant-$r$ photon surface $S_r$ satisfies the condition, $V=V'=0$, which gives the condition for constant-$r$ photon surface, Eq.~(\ref{eq:const-r-psf}), as expected.
The orbit is stable if $V''>0$ and unstable if $V''<0$ at the radius.
Together with the condition, Eq.~(\ref{eq:const-r-psf}), we find the stability condition of the constant-$r$ photon surface $S_r$, Eq.~(\ref{eq:stability}).
\end{proof}
\par
We see examples of constant-$r$ photon surfaces in well-known spacetimes in Appendix~\ref{app:examples}.
We also investigate their stability using the formula in Proposition~\ref{proposition:stability}.

%%%sec
%\section{Examples}
%\llaabel{sec:examples}
%\red{(from the note "PSf of constant-r in D-dim". cut? for known one like Schwarzschild case? to appendix)}

%%%sec
\section{Accretion problem}
\llaabel{sec:accretion}
We consider stationary radial fluid flow on the spacetime~(\ref{eq:metric}).
The flow is spherically, planar or hyperbolically symmetric depending on the spatial symmetry of the spacetime, Eq.~(\ref{eq:maxsym2space}).
After formulating the accretion problem, we give a general analysis using dynamical systems method.
(We sometimes call the problem ``accretion problem" conventionally, however, it would not make sense because the spacetime we consider is not interpreted as spacetime made by some central object in the nonspherical cases.)
The method clarifies the difference between a critical point, which is a singular point of the system, and a sonic point, which is a point where Mach number of flow equals to one.
\par
We assume three conservation equations, the first law, continuity equation and energy-momentum conservation with perfect fluid:
\begin{subnumcases}
{}
\llaabel{eq:firstlaw}
dh=Tds+n^{-1}dp\\
\llaabel{eq:continuity}
\nabla_a J^a=0\\
\llaabel{eq:emconservation}
\nabla_a T^a_b=0
\end{subnumcases}
where $J^a:=nu^a$ is the number current and $T^a_{b}=nhu^a
u_b+p\delta^a_b$ is the energy-momentum tensor of the perfect fluid.
The quantities $h,T,s,n,p$ and $u^a$ are the enthalpy per particle, the temperature,
the entropy per particle, the number density, the pressure and the 4-velocity of the fluid, respectively.
\par
Contraction of Eq.~(\ref{eq:emconservation}) with $u^b$ gives the adiabatic condition of the fluid, $u^a\nabla_as=0$, in general together with Eqs.~(\ref{eq:firstlaw}) and (\ref{eq:continuity}).
The assumption that the fluid is stationary and spherically, planarly or hyperbolically symmetric implies that the entropy is a function of $r$ only, $s=s(r)$, and that the 4-velocity has only the $t$- and $r$- components, $u=u^t\partial_t+u^r\partial_r$.
Besides, the spatial symmetry of the fluid distribution implies constancy of the entropy $s$ on a time slice.
Therefore the entropy $s$ is constant over the whole spacetime and we can write the enthalpy as a function of the number density,
\begin{equation}
\label{eq:adiabaticenthalpy}
h=h(n).
\end{equation}
Integrating Eq.~(\ref{eq:continuity}), we have
\begin{equation}
\label{eq:numberflux}
j_n:=(fg)^{1/2}r^{D-2}nu^r=const,
\end{equation}
from the symmetries of the fluid and the spacetime metric.
The quantity $j_n$ represents the particle flux of the fluid.
Contracting with the static Killing vector $\xi_t^b$, Eq.~(\ref{eq:emconservation}) reduces to the equation of conservation of the energy current $I^a:=nhu_b\xi_t^bu^a$,
\begin{equation}
\nabla_aI^a=0.
\end{equation}
The integration gives
\begin{equation}
\label{eq:energyflux}
j_\epsilon:=(fg)^{1/2}r^{D-2}nhu_tu^r=const.
\end{equation}
where $j_\epsilon$ is the energy flux.
Combining Eqs.~(\ref{eq:adiabaticenthalpy}), (\ref{eq:numberflux}) and (\ref{eq:energyflux}) and using the normalization condition of the 4-velocity $u$,
the problem is formulated into the algebraic master equation:
\begin{equation}
\label{eq:masterequation}
F(r,n):=\left(\frac{j_\epsilon}{j_n}\right)^2
=h^2(n)\left[f(r)+\frac{\mu^2}{r^{2(D-2)}n^2}\right]=const.,
\ \ \ \ \ \mu:=j_n
\end{equation}
$F$ is the energy square per particle and $\mu$ is the parameter interpreted as the accretion rate of the flow.
Our accretion problem is the problem of finding the solution of a fluid flow as a level curve $n=n(r)$ on the phase space $(r,n)$ satisfying $F(r,n(r))=const.$ for a given parameter $\mu$.
%Given the parameter $\mu$, the solution of a fluid flow is obtained as a level curve $n=n(r)$ on the phase space $(r,n)$ satisfying $F(r,n(r))=const.$
Once the number density distribution $n(r)$ obtained for the parameter $\mu$, the equation $j_n:=\left(fg\right)^{1/2}r^{D-2}nu^r=\mu$ gives the corresponding velocity distribution $u^r(r)$.
%Once the number density distribution $n(r)$ obtained for a specific value $\mu$, the equation $j_n=\mu$ gives the corresponding velocity distribution $u^r(r)$ from Eq.~(\ref{eq:numberflux}).
\par
%%%%%
Note that there is no distinction concerning the spatial geometry $s(\chi)$ of the spacetime in the master equation Eq.~(\ref{eq:masterequation}).
Thus far the problem completely coincides with the accretion problem of spherical flow in~\cite{koga}.
We analyze our accretion problem in exactly the same procedure as~\cite{koga} in the following.
%%%subsec
\subsection{Critical point}
Here we give the definition of the critical point and its classification by reformulating the accretion problem in terms of a dynamical system on the phase space $(r,n)$.
The analysis of this kind was first introduced into an accretion problem by Chaverra and Sarbach~\cite{chaverra}.
Generally, the critical point plays an important role in accretion problems and is closely related to the sonic point of the flow.
%%%subsubsec
\subsubsection{Definition of critical point}
In our accretion problem Eq.~(\ref{eq:masterequation}), the solutions are described as level curves of the function $F(r,n)$ on the phase space $(r,n)$.
These curves can be also obtained by integrating the ordinary differential equation,
\begin{equation}
\label{eq:masterode}
\frac{d}{d\lambda}\left(\begin{array}{c} r\\ n \end{array}\right)
=\left(\begin{array}{r} \partial_n \\ -\partial_r \end{array}\right)F(r,n),
\end{equation}
as orbits with a parameter $\lambda$.
This is a reformulation of the master equation Eq.~(\ref{eq:masterequation}) in terms of a dynamical system with the right-hand side (RHS) being the Hamiltonian vector field with respect to the Hamiltonian $F(r,n)$.
Then, the notion of a {\it critical point} (or stationary point as in a dynamical system) at which the RHS of Eq.~(\ref{eq:masterode}) vanishes arises
and its conditions are
\begin{subnumcases}
{}
\llaabel{eq:criticaln}
\partial_nF(r,n)=0\\
\llaabel{eq:criticalr}
\partial_rF(r,n)=0.
\end{subnumcases}
We define a {\it critical point} $(r_c,n_c)$ of the accretion problem as a point on the phase space $(r,n)$ at which the conditions Eqs.~(\ref{eq:criticaln})-(\ref{eq:criticalr}) are satisfied.
%%%subsubsec
\subsubsection{Types of critical points}
\llaabel{sec:types}
The linearization of Eq.~(\ref{eq:masterode}) around a critical point allows us to classify the critical point into two types.
The one is a saddle point and the another one is an extremum point.
A saddle point is a point, in this case, through which two solution orbits pass.
On the other hand, orbits in the vicinity of an extremum point are closed curves around the point.
\par
The linearization matrix $M_c$ is given by
\begin{equation}
M_c:=\left(
	\begin{array}{cc}
	\partial_r\partial_n & \partial_n^2 \\
	-\partial_r^2 & -\partial_r\partial_n 
	\end{array}
	\right)F(r_c,n_c).
\end{equation}
This matrix, being real, $2\times 2$ and traceless, has two eigenvalues with opposite signs.
If the determinant of the matrix is negative (positive), the eigenvalues are real (pure imaginary).
As in a dynamical system, the real eigenvalues imply that the critical point is a saddle point.
For the imaginary eigenvalues, the orbits around the critical point are periodic in linear order.
However, because they are the level curves of the real function $F(r,n)$, the orbits must be closed loops in the vicinity of the critical point.
Therefore the imaginary eigenvalues imply an extremum point.
\par
We know the explicit form of the determinant $\det M_c$ from~\cite{koga}:
\begin{equation}
\det M_c=-\frac{2}{D-2}r_c(f'_c)^2\frac{h_c^4}{n_c^2}\mathcal{F}'(r_c)
\end{equation}
where
\begin{eqnarray}
\mathcal{F}(r)&:=&v_s^2(\bar{n}(r))\left[1+2(D-2)a(r)\right]-1,\nonumber\\
%\bar{n}(r)&:=&\sqrt{\frac{D-2}{2}}\frac{2|\mu|}{\sqrt{r^{2D-3}f'(r)}},\nonumber\\
\bar{n}(r)&:=&|\mu|\sqrt{\frac{2(D-2)}{r^{2D-3}f'(r)}},\nonumber\\
a(r)&:=&\frac{f(r)}{rf'(r)}.\nonumber
\end{eqnarray}
The subscript $c$ means the value at $(r_c,n_c)$.
Then classification of a critical point at radius $r_c$ is given by
\begin{equation}
\llaabel{eq:saddleextremum}
saddle\ (extremum)\ point \Leftrightarrow {\mathcal{F}}'(r_c)>0\ (<0)
\end{equation}
while the critical point $(r_c,n_c)$ itself is also obtained from
\begin{equation}
\label{eq:criticalcondition}
\mathcal{F}(r_c)=0,\;n_c=\bar{n}(r_c).
\end{equation}
%%%subsec
\subsection{Sonic point}
In an accretion problem, the fluid flow can transit from subsonic state (i.e. state where its 3-velocity is smaller than its local sound speed $v_s$) to supersonic state (i.e. state where the 3-velocity is greater than $v_s$) and vice versa.
Such a fluid flow is said to be {\it transonic} and here we call any flow which has both subsonic and supersonic regions {\it transonic flow}.
The point at which transition between subsonic and supersonic states of a transonic flow occurs is called {\it sonic point}.
A sonic point is also related to a critical point mathematically.
%%%subsubsec
\subsubsection{Definition of sonic point}
Since, in our accretion problem, a fluid accretion flow is a solution orbit of Eq.~(\ref{eq:masterequation}),
we define a {\it sonic point} of a transonic flow as a point on the phase space $(r,n)$:
\begin{definition}
\label{definition:sonicpoint}
For a transonic fluid flow of our accretion problem Eq.~(\ref{eq:masterequation}), let $n=n(r)$ be the corresponding solution curve on the phase space $(r,n)$.
Let $v=v(r)$ be the 3-velocity of the flow at radius $r$ measured by static observers.
A sonic point $(r_s,n_s)$ of the flow is a point on the phase space satisfying the condition
\begin{eqnarray}
\label{eq:sonicpoint}
\left.\frac{v^2}{v_s^2}\right|_{(r_s,n(r_s))}=1,
\end{eqnarray}
where $n_s=n(r_s)$.
\end{definition}
%%%subsubsec
\subsubsection{Sonic point and critical point}
The static observer $u_{o}=f^{-1/2}\partial_t$ measures the squared fluid 3-velocity $v^2$ by
\begin{equation}
\frac{1}{1-v^2}=\left(g_{ab}u^au_o^b\right)^2
\end{equation}
which gives
\begin{equation}
\label{eq:vsquared}
v^2=\frac{\mu^2}{\mu^2+fr^{2(D-2)}n^2}.
\end{equation}
Let us calculate explicitly one of the conditions for the critical point Eq.~(\ref{eq:criticaln}):
\begin{eqnarray}
0&=&\partial_nF\nonumber\\
&=&\frac{2h^2}{n}\frac{\mu^2}{r^{2(D-2)} n^2}\left(v_s^2(n)\left[1+f\frac{r^{2(D-2)} n^2}{\mu^2}\right]-1\right)
\end{eqnarray}
where $v_s^2(n):=\partial\ln h/\partial \ln n$ is the sound speed.
We can see that the sound speed can be always written as
\begin{equation}
\label{eq:vssquared}
v_s^2(n)=\frac{\mu^2}{\mu^2+fr^{2(D-2)} n^2}
\end{equation}
on the points $(r,n)$ satisfying the condition Eq.~(\ref{eq:criticaln}) including the critical point.
Conversely, if Eq.~(\ref{eq:vssquared}) is satisfied on a given point $(r,n)$, the condition Eq.~(\ref{eq:criticaln}) holds.
From this fact and Eqs.~(\ref{eq:vsquared}) and~(\ref{eq:vssquared}), we can show that a sonic point of a physically acceptable transonic flow is identified with a critical point of saddle type as follows.
\par
Consider a physical transonic fluid flow which is specified by the solution curve $n=n(r)$.
From Definition~\ref{definition:sonicpoint} and Eq.~(\ref{eq:vsquared}), the sonic point $(r_s,n_s)$ is determined by
\begin{equation}
\llaabel{eq:spcondition}
v_s^2\left(n(r_s)\right)=\frac{\mu^2}{\mu^2+fr^{2(D-2)}n^2(r_s)}
\end{equation}
and $n_s=n(r_s)$.
Since Eq.~(\ref{eq:spcondition}) implies that the point $(r_s,n_s)$ satisfies Eq.~(\ref{eq:vssquared}), the condition $\partial_n F(r_s,n_s)=0$ also holds as mentioned below Eq.~(\ref{eq:vssquared}).
Then we have the following three cases concerning the sonic point $(r_s,n_s)$:
\begin{enumerate}
\item $\partial_r F(r_s,n_s)\neq0$ (i.e., the sonic point is not a critical point). \label{cond:non-critical}
\item $\partial_r F(r_s,n_s)=0$ (i.e., the sonic point is a critical point due to the fact $\partial_n F(r_s,n_s)=0$).
\begin{enumerate}
\item The corresponding critical point is of saddle type.  \label{cond:saddle}
\item The corresponding critical point is of extremum type. \label{cond:extremum}
\end{enumerate}
\end{enumerate}
In the case~\ref{cond:non-critical}, the curve $n=n(r)$ typically gets double-valued  (so unphysical) at least around $(r_s,n_s)$ locally because $dn/dr=-\partial_rF(r_s,n_s)/\partial_nF(r_s,n_s)=\pm\infty$ there from
Eq.~$(\ref{eq:masterode})$. 
Another possibility with diverging density gradient, which is physically acceptable, is a transonic shock.
In the current paper, we require the finite density gradient, $|dn/dr|<\infty$, as one of the conditions of a physical flow, thus excluding a transonic shock.
Therefore the case~\ref{cond:non-critical} is not allowed for the physical flow $n=n(r)$ and the sonic point must be a critical point.
However, the case~\ref{cond:extremum} is also excluded because any
solution curve, being a level curve of $F(r,n)$ originally, cannot pass the critical point of extremum type.
Then we have only the case~\ref{cond:saddle} for the sonic point $(r_s,n_s)$ of the physically acceptable transonic flow $n=n(r)$.
As a consequence, we have the following theorem:
\begin{theorem}
\llaabel{theorem:sonic-critical}
For a physical transonic fluid flow which is stationary and spherically, planar or hyperbolically symmetric on the spacetime~(\ref{eq:metric}), its sonic point coincides with a critical point of saddle type on the phase space.
\end{theorem}
We can interpret a critical point of saddle type as a sonic point of some transonic flow which is physically acceptable at least in the vicinity of the point on the phase space $(r,n)$.

%%%sec
\section{SP/PS correspondence}
\llaabel{sec:correspondence}
In this section, we analyze a critical point of radiation fluid flow and, as the main result of the current paper, prove Theorem~\ref{theorem:correspondence}.
\if0
\begin{theorem}
\llaabel{theorem:correspondence}
For any physical transonic flow of ideal photon gas in stationary and spherically, planar or hyperbolically symmetric state on the spacetime~(\ref{eq:metric}), the radius of its sonic point coincides with that of (one of) the unstable constant-$r$ photon surface(s).
%For a physical transonic flow of ideal photon gas of our accretion problem Eq.~(\ref{eq:masterequation}), its sonic point(s) must be on (one of) the unstable constant-$r$ photon surface(s) of the spacetime Eq.~(\ref{eq:metric}).
\end{theorem}
\fi
%%%subsec
\subsection{Critical point of radiation fluid flow}
We derived the EOS of radiation (ideal photon gas) in arbitrary spatial dimensions $d$ in~\cite{koga}.
For our purpose, it is sufficient to know that the enthalpy can be written in the form,
\begin{equation}
h(n)=(const.)\times n^{\gamma-1},
\end{equation}
where the index $\gamma$ is related to the dimension by $\gamma=(d+1)/d$.
The sound speed $v_s^2(n)$ is then computed as
\begin{equation}
\llaabel{vs-radiation}
v_s^{2}(n):=\frac{\partial \ln h}{\partial \ln n}=\gamma-1=\frac{1}{d}=\frac{1}{D-1}.
\end{equation}
\par
For the conditions of the critical point $(r_c,n_c)$ and its classification for radiation flow, we have the following lemma:
\begin{lemma}
\llaabel{lemma:photongascritical}
For radiation fluid flow in our accretion problem Eq.~(\ref{eq:masterequation}), radius $r_c$ of a critical point is specified by
\begin{equation}
\llaabel{eq:criticalradiusradiation}
(fr^{-2})'=0
\end{equation}
and the corresponding critical density $n_c$ is
\begin{equation}
\llaabel{eq:criticalnumberradiation}
n_c=\left|\mu\right|\sqrt{\frac{2(D-2)}{r_c^{2D-3}f'_c}}.
\end{equation}
The type of the critical point is classified by the inequality,
\begin{equation}
\llaabel{eq:classificationradiation}
saddle\ (extremum)\ point \Leftrightarrow \left(fr^{-2}\right)''_{r=r_c}<0\ (>0).
\end{equation}
\end{lemma}
\begin{proof}
Substituting the sound speed of radiation fluid, Eq.~(\ref{vs-radiation}), into Eq.~(\ref{eq:criticalcondition}), the condition for the critical radius $r_c$ is given by
\begin{equation}
\mathcal{F}=-\frac{D-2}{D-1}\frac{1}{f'r^{-2}}\left(fr^{-2}\right)'=0.
\end{equation}
Therefore, the critical radius is specified by Eq.~(\ref{eq:criticalradiusradiation}).
%\begin{equation}
%\llaabel{eq:criticalradiusforradiation}
%\left(fr^{-2}\right)'=0.
%\end{equation}
Once the radius $r_c$ is obtained, we get the corresponding number density $n_c$ from Eq.~(\ref{eq:criticalcondition}) which gives Eq.~(\ref{eq:criticalnumberradiation}).
%\begin{equation}
%n_c=\left|\mu\right|\sqrt{\frac{2(D-2)}{r_c^{2D-3}f'_c}}.
%\end{equation}
With the use of Eq.~(\ref{eq:criticalradiusradiation}), the left-hand side (LHS) of the classification condition of a critical point, Eq.~(\ref{eq:saddleextremum}), is written as
\begin{equation}
\mathcal{F}'(r_c)=-\frac{D-2}{D-1}\frac{1}{f'r^{-2}}\left(fr^{-2}\right)''
\end{equation}
and we immediately obtain Eq.~(\ref{eq:classificationradiation}).
Note that $f’=2f/r>0$ at $r=r_c$ from Eq.~(\ref{eq:criticalradiusradiation}).
\end{proof}
%%%subsec
\subsection{Proof of theorem: SP/PS correspondence}
The conditions of the critical radius $r_c$ and its classification in Lemma coincide with the conditions of the radius of the constant-$r$ photon surface and its stability in Proposition~\ref{proposition:const-r-psf} and~\ref{proposition:stability}, respectively.
Then we immediately obtain the following proposition about the correspondence between constant-$r$ photon surfaces and critical points of radiation fluid flow:
\begin{proposition}
\llaabel{proposition:cp-ps}
A critical point of radiation fluid flow in our accretion problem~(\ref{eq:masterequation}) exists at a radius $r$ if and only if the spacetime Eq.~(\ref{eq:metric}) has a constant-$r$ photon surface $S_r$ at the radius.
If the constant-$r$ photon surface $S_r$ is stable, the critical point is of a saddle type while if unstable, it is of an extremum type.
\end{proposition}
There is a one-to-one correspondence between critical points of radiation fluid flow and constant-$r$ photon surfaces.
It is worth noting that if the spacetime has more than one constant-$r$ photon surfaces, the stable and unstable photon surfaces appear alternately as we can see from Eqs.~(\ref{eq:criticalradiusradiation}) and~(\ref{eq:classificationradiation}).
The fact also leads to the alternate appearance of the corresponding extremum and saddle points on the phase space $(r,n)$.
\par
Then Theorem~\ref{theorem:sonic-critical} and Proposition~\ref{proposition:cp-ps} above immediately prove Theorem~\ref{theorem:correspondence}.
\par
For a sonic point $(r_c,n_c)$, the surface of the constant radius $r_c$ is sometimes referred to as {\it sonic surface}.
In this view point, Theorem~\ref{theorem:correspondence} also states that the sonic surface coincides with (one of) the photon surface(s).

%%%sec
\section{Conclusion}
\llaabel{sec:conclusion}
We investigated photon surfaces of constant radius in the spacetime~(\ref{eq:metric}) and defined its stability in the same way as we defined the stability of a photon sphere in~\cite{koga}.
In spite of the different spatial geometries, spherical, planar and hyperbolic symmetry, their conditions in Propositions~\ref{proposition:const-r-psf} and~\ref{proposition:stability} were found to be exactly the same.
In other words, they are independent of the function $s(\chi)$ which depends on the spatial symmetry of the spacetime according to Eq.~(\ref{eq:maxsym2space}).
This fact comes from that a photon surface is a structure of spacetime characterized by its second fundamental form and, in the cases we investigated, the second fundamental forms of the hypersurfaces take the same form irrelevant to the spatial symmetry in the tetrad system.
\par
We formulated the accretion problem of stationary radial fluid flow which is also spatially symmetric depending on the spatial symmetry of the spacetime, Eq.~(\ref{eq:maxsym2space}).
It was revealed that the master equation~(\ref{eq:masterequation}) does not depend on the spatial symmetry explicitly.
Therefore we applied the dynamical system analysis to the problem and obtained the same results about the critical points and the sonic points as the spherical case~\cite{koga}.
\par
Together with the results of the photon surfaces and the sonic points, we proved the main theorem of the current paper, Theorem~\ref{theorem:correspondence}, which states the correspondence between sonic points of radiation fluid flow and {\it photon surfaces}.
This is the extension of the theorem for SP/PS correspondence in spherically symmetric spacetime~\cite{koga} to non-spherically symmetric spacetime of the same degrees of symmetry.
\par
The main theorem has many implications which answers our questions about SP/PS correspondence.
In the previous papers~\cite{koga}~\cite{koga2}, we found that there always exits a correspondence between sonic points and photon spheres in the spherically symmetric spacetime for the radial and rotational flow.
However, we did not know which the aspects of photon spheres are responsible to the correspondence.
In this paper, from the correspondence between sonic points and {\it photon surfaces}, we can conclude that the umbilicity of the hypersurfaces plays the most important role, or possibly, is needed in the correspondence.
Furthermore, the sphericity of the surface, the positivity of the intrinsic curvature and the closed spatial topology are not necessary for the correspondence.
We can also infer that the correspondence will occur caused by the local geometrical structure and the extrinsic structure of the hypersurfaces rather than the global structure and the intrinsic structure.
Since the correspondence seems to originate from the microscopic construction of radiation fluid, our new result may suggest that transonic fluid flows generally have their sonic points only at points where the geometry has some special structures and the special structures are characterized by the geodesic motions of the particles which constitute the fluid.
\par
Since sonic points, or sonic surfaces, coincide with photon surfaces rather than photon spheres, it would be better to read ``SP/PS correspondence" as ``Sonic point/{\it Photon surface} correspondence".

\begin{acknowledgments}
The author thanks F. C. Mena, T. Harada, S. Jahns, Y. Nakayama, T. Igata, K. Ogasawara, T. Kokubu and K. Nakashi for their very helpful discussions and comments.
This work was supported by Rikkyo University Special Fund for Research.
%This work was...Grant...\red{SFR?}.
%This work was partially supported by JSPS KAKENHI Grant No. JP26400282 (T.H.).
\end{acknowledgments}

\appendix
\section{Examples of constant-$r$ photon surface}
\llaabel{app:examples}
We see the examples of constant-$r$ photon surface in the following.
%%%subsec
\subsection{Schwarzschild spacetime}
Schwarzschild spacetime is given by the condition,
\begin{equation}
f(r)=g^{-1}(r)=1-\frac{2M}{r},\;s(\chi)=\sin\chi,\;D=4,
\end{equation}
for the metric Eq.~(\ref{eq:metric}).
From Proposition~\ref{proposition:const-r-psf}, there exists a unique photon surface (sphere) in Schwarzschild spacetime and it is located on $r=3M$.
From Proposition~\ref{proposition:stability}, the stability condition at the radius,
\begin{equation}
(fr^{-2})''|_{r=3M}=-\frac{2}{81M^4}<0,
\end{equation}
implies the photon surface is unstable.
This is the well-known result.
%%%subsec
\subsection{C-metric}
Consider the specific case of C-metric given by
\begin{equation}
\llaabel{eq:c-metric}
f(r)=g^{-1}(r)=-\frac{k}{l^2}r^2+b-\frac{2m}{r}+\frac{q^2}{r^2},\;D=4
\end{equation}
where %$m>0$ and 
$0<r<\infty$.
Without loss of generality, the magnitude of the parameters $b$ and $k$ can be set to $1$ if they are nonzero.
From Proposition~\ref{proposition:const-r-psf}, the umbilicity condition of a timelike hypersurface $S_r$ is given by
\begin{equation}
\llaabel{eq:c-psf}
%(fr^{-2})'=-2r^{-5}\left(br^2-3mr+2q^2\right)=0
br^2-3mr+2q^2=0
\end{equation}
for $r$ satisfying $f(r)>0$.
From Proposition~\ref{proposition:stability}, its stability is determined by
\begin{equation}
\llaabel{eq:c-stability}
%stable(unstable)\Leftrightarrow\left.(fr^{-2})''\right|_{r=r_{ph}}=-2r^{-4}\left(2br-3m\right)>\left(<0\right)
stable(unstable)\Leftrightarrow2br_{ph}-3m<0\;\left(>0\right)
\end{equation}
where $r_{ph}$ is the radius of the const-$r$ photon surface.
Note that for $r=r_{ph}$ surface to be a timelike photon surface, it must satisfy $f(r_{ph})>0$, i.e. it must be outside the horizons.
\par
Vacuum Einstein-Maxwell equation implies the following three cases in the presence of the cosmological constant~\cite{mann}:
\begin{equation}
\label{eq:c-metricCases}
\left\{
 \begin{array}{l}
 b=+1,\ k=\pm1,\  s(\chi)=\sin\chi\\
 b=0,\ k=-1,\ s(\chi)=1\\
 b=-1,\ k=-1,\ s(\chi)=\sinh\chi.
 \end{array}
\right.
\end{equation}
$k$ and $l^2$ are related to the cosmological constant $\Lambda$ by $\Lambda=3k/l^2$.
$q^2$ is the charge with the corresponding gauge field,
\begin{equation}
F_M=qs(\chi)d\chi\land d\phi,\ F_E=-\frac{q}{r^2}dt\land dr.
\end{equation}
This is the solution of vacuum Einstein-Maxwell equation with the cosmological constant for the metric ansatz, Eq.~(\ref{eq:metric}).
\par
In the spherical case $b=1$ and $k=\pm1$, Eq.~(\ref{eq:c-psf}) has solutions for $0<r<\infty$ only if $9m^2-8q^2\ge0$ and $m>0$.
They are
\begin{equation}
r_{ph\pm}=\frac{3m\pm\sqrt{9m^2-8q^2}}{2}.
\end{equation}
The outer one is unstable and the inner one is stable from Eq.~(\ref{eq:c-stability}).
When the equality holds, $9m^2-8q^2=0$, the two photon surfaces coincide, $r_{ph+}=r_{ph-}$.
The marginally case implies that the circular orbits are on the inflection point of their potentials $V(r)$, i.e. $V'(r_{ph})=V''(r_{ph})=0$.
As mentioned above, the radii must be outside the horizons so that the hypersurfaces of $r=r_{ph\pm}$ are photon surfaces.
Consider $k=0$ case for example.
The spacetime is then Reissner–Nordstr\"{o}m spacetime.
If the spacetime is over extremal ($q^2>m^2$), there are no horizons and the surfaces of $r=r_{ph\pm}$ are indeed photon surfaces.
If the spacetime is sub extremal ($q^2<m^2$), we can see that $f(r_{ph+})>0$ and $f(r_{ph-})<0$ from short calculations provided $9m^2-8q^2\ge0$.
Therefore, only the outer radius $r_{ph+}$ is a timelike photon surface and the radii have relation $r_{h-}<r_{ph-}<r_{h+}<r_{ph+}$ where $r_{h\pm}$ are the outer and inner horizons, respectively.
\par
In the planar case $b=0$ and $k=-1$, Eq.~(\ref{eq:c-psf}) has a solution only if $m>0$.
The radius is uniquely given by
\begin{equation}
r_{ph}=\frac{2q^2}{3m}
\end{equation}
and the photon surface is stable from Eq.~(\ref{eq:c-stability}).
If the spacetime has no horizons, the radius gives a photon surface.
However, if there are horizons, we can see that the radius is inside the horizons as follows:
If there exists some radius $r$ such that $\tilde{f}(r):=f(r)r^2\le0$, the spacetime has horizons and otherwise no horizons because $\lim_{r\to0}\tilde{f}(r)>0$ and $\lim_{r\to\infty}\tilde{f}(r)>0$.
The polynomial $\tilde{f}(r)$ has the minimum value $\tilde{f}_{min}=\tilde{f}(r_{min})$ at $r=r_{min}=(l^2m/2)^{1/3}$.
The horizons exist if and only if $\tilde{f}_{min}\le0$ and calculation of $\tilde{f}_{min}$ reveals that this is equivalent to the condition $q^6/l^2m^4\le27/16$.
Using the fact, we have $\tilde{f}(r_{ph})=16q^2/81\left(q^6/l^2m^4-27/16\right)\le0$ if the horizons exist.
Therefore $f(r_{ph})\le0$ and there are no constant-$r$ photon surface outside the horizons if the spacetime has horizons.
\par
In the hyperbolic case $b=-1$ and $k=-1$, only the allowed solution of Eq.~(\ref{eq:c-psf}) for $0<r<\infty$ is
\begin{equation}
r_{ph}=\frac{\sqrt{9m^2+8q^2}-3m}{2}.
\end{equation}
The stability condition Eq.~(\ref{eq:c-stability}) implies the corresponding photon surface is stable irrelevant to whether $m$ is negative or positive.
If $m\ne0$ and $q^2>0$, the spacetime has a critical value $l_c^2$ such that there are no horizons, $f(r)>0\;\forall r$, for $l^2<l_c^2$ because $\left.f(r)\right|_{k=0}$ is bounded below for $0<r<\infty$.
Therefore the surface of $r=r_{ph}$ is a timelike photon surface if, at least, $l^2<l_c^2$.
In the limit $l^2\to\infty$, or $k=0$ case, we can prove that $f(r_{ph})<0$ from straightforward calculations provided $m\ne0$ and $q^2>0$.
There are no constant-$r$ photon surfaces in that case.
\par
Vacuum Einstein-Maxwell equation with the cosmological constant admits, at most, one stable and one unstable  constant-$r$ photon surface.
The unstable one can exist only in the spherical case.
In the absence of the charge, $q^2=0$, there is no stable constant-$r$ photon surface in all the cases of the symmetry.
\par
The paper for the detailed analysis about the existence of constant-$r$ photon surfaces is in preparation~\cite{koga:const-r-psf}.

%%%sec
\section{Christoffel symbol}
\llaabel{app:christoffel}
We used the components of Christoffel symbol below in the calculations in Sec.~\ref{sec:psf-const-r}.
%\red{Not yet corrected that $\sqrt{g}\to g$.}
\begin{eqnarray}
{\Gamma^r}_{tt}&=&\frac{1}{2}g^{rr}\left(2g_{rt,t}-g_{tt,r}\right)=\frac{1}{2}g^{-1}f'\\
{\Gamma^r}_{t\chi}&=&\frac{1}{2}g^{rr}\left(g_{rt,\chi}+g_{r\chi,t}-g_{t\chi,r}\right)=0\\
{\Gamma^r}_{t\theta_1}&=&\frac{1}{2}g^{rr}\left(g_{rt,\theta_1}+g_{r\theta_1,t}-g_{t\theta_1,r}\right)=0\\
{\Gamma^r}_{\chi\chi}&=&\frac{1}{2}g^{rr}\left(2g_{r\chi,\chi}-g_{\chi\chi,r}\right)=-g^{-1}r\\
{\Gamma^r}_{\chi\theta_1}&=&\frac{1}{2}g^{rr}\left(g_{r\chi,\theta_1}+g_{r\theta_1,\chi}-g_{\chi\theta_1,r}\right)=0\\
{\Gamma^r}_{\theta_1\theta_1}&=&\frac{1}{2}g^{rr}\left(2g_{r\theta_1,\theta_1}-g_{\theta_1\theta_1,r}\right)=-g^{-1}s^2r
\end{eqnarray}

%\bibliography{sample}
%
\bibliography{const_r_psf_correspondence}

\begin{thebibliography}{13}%
\makeatletter
\providecommand \@ifxundefined [1]{%
 \@ifx{#1\undefined}
}%
\providecommand \@ifnum [1]{%
 \ifnum #1\expandafter \@firstoftwo
 \else \expandafter \@secondoftwo
 \fi
}%
\providecommand \@ifx [1]{%
 \ifx #1\expandafter \@firstoftwo
 \else \expandafter \@secondoftwo
 \fi
}%
\providecommand \natexlab [1]{#1}%
\providecommand \enquote  [1]{``#1''}%
\providecommand \bibnamefont  [1]{#1}%
\providecommand \bibfnamefont [1]{#1}%
\providecommand \citenamefont [1]{#1}%
\providecommand \href@noop [0]{\@secondoftwo}%
\providecommand \href [0]{\begingroup \@sanitize@url \@href}%
\providecommand \@href[1]{\@@startlink{#1}\@@href}%
\providecommand \@@href[1]{\endgroup#1\@@endlink}%
\providecommand \@sanitize@url [0]{\catcode `\\12\catcode `\$12\catcode
  `\&12\catcode `\#12\catcode `\^12\catcode `\_12\catcode `\%12\relax}%
\providecommand \@@startlink[1]{}%
\providecommand \@@endlink[0]{}%
\providecommand \url  [0]{\begingroup\@sanitize@url \@url }%
\providecommand \@url [1]{\endgroup\@href {#1}{\urlprefix }}%
\providecommand \urlprefix  [0]{URL }%
\providecommand \Eprint [0]{\href }%
\providecommand \doibase [0]{http://dx.doi.org/}%
\providecommand \selectlanguage [0]{\@gobble}%
\providecommand \bibinfo  [0]{\@secondoftwo}%
\providecommand \bibfield  [0]{\@secondoftwo}%
\providecommand \translation [1]{[#1]}%
\providecommand \BibitemOpen [0]{}%
\providecommand \bibitemStop [0]{}%
\providecommand \bibitemNoStop [0]{.\EOS\space}%
\providecommand \EOS [0]{\spacefactor3000\relax}%
\providecommand \BibitemShut  [1]{\csname bibitem#1\endcsname}%
\let\auto@bib@innerbib\@empty
%</preamble>
\bibitem [{\citenamefont {Synge}(1965)}]{synge}%
  \BibitemOpen
  \bibfield  {author} {\bibinfo {author} {\bibfnamefont {J.~L.}\ \bibnamefont
  {Synge}},\ }\href@noop {} {\bibfield  {journal} {\bibinfo  {journal} {Monthly
  Notices Roy Astronom. Soc.}\ }\textbf {\bibinfo {volume} {131}},\ \bibinfo
  {pages} {463} (\bibinfo {year} {1965})}\BibitemShut {NoStop}%
\bibitem [{\citenamefont {Cardoso}\ \emph {et~al.}(2009)\citenamefont
  {Cardoso}, \citenamefont {Miranda}, \citenamefont {Berti}, \citenamefont
  {Witek},\ and\ \citenamefont {Zanchin}}]{cardoso}%
  \BibitemOpen
  \bibfield  {author} {\bibinfo {author} {\bibfnamefont {V.}~\bibnamefont
  {Cardoso}}, \bibinfo {author} {\bibfnamefont {A.~S.}\ \bibnamefont
  {Miranda}}, \bibinfo {author} {\bibfnamefont {E.}~\bibnamefont {Berti}},
  \bibinfo {author} {\bibfnamefont {H.}~\bibnamefont {Witek}}, \ and\ \bibinfo
  {author} {\bibfnamefont {V.~T.}\ \bibnamefont {Zanchin}},\ }\href@noop {}
  {\bibfield  {journal} {\bibinfo  {journal} {Phys. Rev. D}\ }\textbf {\bibinfo
  {volume} {79}},\ \bibinfo {pages} {064016} (\bibinfo {year}
  {2009})}\BibitemShut {NoStop}%
\bibitem [{\citenamefont {Hod}(2009)}]{hod}%
  \BibitemOpen
  \bibfield  {author} {\bibinfo {author} {\bibfnamefont {S.}~\bibnamefont
  {Hod}},\ }\href@noop {} {\bibfield  {journal} {\bibinfo  {journal} {Phys.
  Rev. D}\ }\textbf {\bibinfo {volume} {80}},\ \bibinfo {pages} {064004}
  (\bibinfo {year} {2009})}\BibitemShut {NoStop}%
\bibitem [{\citenamefont {Claudel}\ \emph {et~al.}(2001)\citenamefont
  {Claudel}, \citenamefont {Virbhadra},\ and\ \citenamefont {Ellis}}]{claudel}%
  \BibitemOpen
  \bibfield  {author} {\bibinfo {author} {\bibfnamefont {C.~M.}\ \bibnamefont
  {Claudel}}, \bibinfo {author} {\bibfnamefont {K.~S.}\ \bibnamefont
  {Virbhadra}}, \ and\ \bibinfo {author} {\bibfnamefont {G.~F.~R.}\
  \bibnamefont {Ellis}},\ }\href@noop {} {\bibfield  {journal} {\bibinfo
  {journal} {J. Math. Phys.}\ }\textbf {\bibinfo {volume} {42}},\ \bibinfo
  {pages} {818} (\bibinfo {year} {2001})}\BibitemShut {NoStop}%
\bibitem [{\citenamefont {Bondi}(1952)}]{bondi}%
  \BibitemOpen
  \bibfield  {author} {\bibinfo {author} {\bibfnamefont {H.}~\bibnamefont
  {Bondi}},\ }\href@noop {} {\bibfield  {journal} {\bibinfo  {journal} {Monthly
  Notices Roy Astronom. Soc.}\ }\textbf {\bibinfo {volume} {112}},\ \bibinfo
  {pages} {195} (\bibinfo {year} {1952})}\BibitemShut {NoStop}%
\bibitem [{\citenamefont {Michel}(1972)}]{michel}%
  \BibitemOpen
  \bibfield  {author} {\bibinfo {author} {\bibfnamefont {F.}~\bibnamefont
  {Michel}},\ }\href@noop {} {\bibfield  {journal} {\bibinfo  {journal}
  {Astrophysics and Space Science}\ }\textbf {\bibinfo {volume} {15}},\
  \bibinfo {pages} {153} (\bibinfo {year} {1972})}\BibitemShut {NoStop}%
\bibitem [{\citenamefont {Mach}\ \emph {et~al.}(2013)\citenamefont {Mach},
  \citenamefont {Malec},\ and\ \citenamefont {Karkowski}}]{mach}%
  \BibitemOpen
  \bibfield  {author} {\bibinfo {author} {\bibfnamefont {P.}~\bibnamefont
  {Mach}}, \bibinfo {author} {\bibfnamefont {E.}~\bibnamefont {Malec}}, \ and\
  \bibinfo {author} {\bibfnamefont {J.}~\bibnamefont {Karkowski}},\ }\href@noop
  {} {\bibfield  {journal} {\bibinfo  {journal} {Phys. Rev. D}\ }\textbf
  {\bibinfo {volume} {88}},\ \bibinfo {pages} {084056} (\bibinfo {year}
  {2013})}\BibitemShut {NoStop}%
\bibitem [{\citenamefont {Chaverra}\ and\ \citenamefont
  {Sarbach}(2015)}]{chaverra}%
  \BibitemOpen
  \bibfield  {author} {\bibinfo {author} {\bibfnamefont {E.}~\bibnamefont
  {Chaverra}}\ and\ \bibinfo {author} {\bibfnamefont {O.}~\bibnamefont
  {Sarbach}},\ }\href@noop {} {\bibfield  {journal} {\bibinfo  {journal}
  {Class. Quantum Grav.}\ }\textbf {\bibinfo {volume} {32}},\ \bibinfo {pages}
  {155006} (\bibinfo {year} {2015})}\BibitemShut {NoStop}%
\bibitem [{\citenamefont {Koga}\ and\ \citenamefont {Harada}(2016)}]{koga}%
  \BibitemOpen
  \bibfield  {author} {\bibinfo {author} {\bibfnamefont {Y.}~\bibnamefont
  {Koga}}\ and\ \bibinfo {author} {\bibfnamefont {T.}~\bibnamefont {Harada}},\
  }\href@noop {} {\bibfield  {journal} {\bibinfo  {journal} {Phys. Rev. D}\
  }\textbf {\bibinfo {volume} {94}},\ \bibinfo {pages} {044053} (\bibinfo
  {year} {2016})},\ \Eprint {http://arxiv.org/abs/1601.07290} {arXiv:1601.07290
  [gr-qc]} \BibitemShut {NoStop}%
\bibitem [{\citenamefont {Koga}\ and\ \citenamefont {Harada}(2018)}]{koga2}%
  \BibitemOpen
  \bibfield  {author} {\bibinfo {author} {\bibfnamefont {Y.}~\bibnamefont
  {Koga}}\ and\ \bibinfo {author} {\bibfnamefont {T.}~\bibnamefont {Harada}},\
  }\href@noop {} {\bibfield  {journal} {\bibinfo  {journal} {Phys. Rev. D}\
  }\textbf {\bibinfo {volume} {98}},\ \bibinfo {pages} {024018} (\bibinfo
  {year} {2018})},\ \Eprint {http://arxiv.org/abs/1803.06486} {arXiv:1803.06486
  [gr-qc]} \BibitemShut {NoStop}%
\bibitem [{\citenamefont {Perlick}(2005)}]{perlick}%
  \BibitemOpen
  \bibfield  {author} {\bibinfo {author} {\bibfnamefont {V.}~\bibnamefont
  {Perlick}},\ }\href@noop {} {\bibfield  {journal} {\bibinfo  {journal}
  {Nonlinear Analysis}\ }\textbf {\bibinfo {volume} {63/5-7}},\ \bibinfo
  {pages} {e511–e518} (\bibinfo {year} {2005})}\BibitemShut {NoStop}%
\bibitem [{\citenamefont {Mann}(1997)}]{mann}%
  \BibitemOpen
  \bibfield  {author} {\bibinfo {author} {\bibfnamefont {R.~B.}\ \bibnamefont
  {Mann}},\ }\href@noop {} {\bibfield  {journal} {\bibinfo  {journal} {Class.
  Quant. Grav.}\ }\textbf {\bibinfo {volume} {14}},\ \bibinfo {pages} {L109}
  (\bibinfo {year} {1997})}\BibitemShut {NoStop}%
\bibitem [{\citenamefont {Koga}()}]{koga:const-r-psf}%
  \BibitemOpen
  \bibfield  {author} {\bibinfo {author} {\bibfnamefont {Y.}~\bibnamefont
  {Koga}},\ }\href@noop {} {\bibinfo  {journal} {(work in progress)}\
  }\BibitemShut {NoStop}%
\end{thebibliography}%

\end{document}